\documentclass[journal, onecolumn, 11pt]{IEEEtran}
\usepackage{graphicx,amsthm}
\usepackage[cmex10]{amsmath}

\newtheorem{theorem}{Theorem}
\newtheorem{definition}[theorem]{Definition}
\newtheorem{lemma}[theorem]{Lemma}

\newcommand\dist[1]{\langle#1\rangle}
\newcommand\ceiling[1]{\left\lceil #1 \right\rceil}

\newcommand\bp{\mathbf{p}}

\newcommand\dpr[1]{(\kern-.1977em(#1)\kern-.1977em)}

\title{Binarization Trees and Random Number Generation} 
\author{Sung-il Pae\\
  \thanks{Sung-il Pae is with Department of Computer Engineering, Hongik
    University, Seoul, Korea. (email: pae@hongik.ac.kr)
This work was presented in part at 2016 IEEE Symposium on Information Theory
(ISIT 2016), July 10-16, 2016, Barcelona, Spain.
    This research was supported in part by a Hongik University grant
    and the National Research Foundation of Korea (NRF) grant funded by the
    Korean government (No. 2016R1D1A1B01016531). 
}}

\begin{document}
\maketitle

\begin{abstract}
\noindent An {\em $m$-extracting procedure\/} produces unbiased random bits
from a loaded dice with $m$ faces.  A {\em binarization\/} takes inputs from
an $m$-faced dice and produce bit sequences to be fed into a (binary)
extracting procedure to obtain random bits.  Thus, binary extracting
procedures give rise to an $m$-extracting procedure via a binarization.  An
entropy-preserving binarization is to be called {\em complete,\/} and such a
procedure has been proposed by Zhou and Bruck.  We show that there exist
complete binarizations in abundance as naturally arising from binary trees
with $m$ leaves.  The well-known {\em leaf entropy theorem\/} and a closely
related structure lemma play important roles in the arguments.
\end{abstract}


\section{Introduction}
An $m$-{\em extracting\/} procedure produces unbiased random bits using a
sequence from an i.i.d. source over an alphabet $\{0,1,\dots,m-1\}$,
regardless of its probability distribution $\dist{p_0,p_1,\dots,p_{m-1}}$.
When $m=2$, the source is a biased coin, and the famous von Neumann trick is
2-extracting: take a pair of coin flips and return random bits by the
following rule~\cite{VonNeumann51}:
\begin{equation}\label{eq:vonNeumann}
00\mapsto \lambda,\; 01\mapsto 0,\; 10\mapsto 1,\; 11\mapsto \lambda,
\end{equation}
where $\lambda$ indicates ``no output.''  Because $\Pr(01)=\Pr(10)=p_0p_1$,
the resulting bit is unbiased, and the {\em output rate}, the average number of
output per input, is $p_0p_1\le 1/4$.  Elias~\cite{elias72} and
Peres~\cite{peres92} extend it by taking inputs of length $n\ge2$ and
returning more than one bit at a time.  Both methods are asymptotically
optimal; as the input size $n$ increases, the output rate approaches the
information-theoretic upper bound $H(p_0)$, the Shannon
entropy~\cite{shannon64mathematical,DBLP:books/daglib/0016881}.

Elias's method generalizes naturally from 2-extracting to $m$-extracting
procedures for each $m>2$, as discussed in Elias's original
paper~\cite{elias72}.  However, a similar generalization of Peres's method
had been unknown for quite a while and was found only recently~\cite{pae15}.
In the meanwhile, Zhou and Bruck proposed a very interesting scheme that
transforms any binary extracting procedure into an $m$-extracting
procedure~\cite{DBLP:journals/corr/abs-1209-0726}.  For example, Peres method
is turned into an $m$-extracting procedure via a simple process called
``binarization.''  If the above-mentioned generalizations of Elias and Peres
are to be called direct generalizations, their scheme is rather a
meta-generalization.  Moreover, the resulting $m$-extracting procedure is
claimed to be asymptotically optimal if the given 2-extracting procedure is
asymptotically optimal.

In this paper, such entropy-preserving processes will be called
{\em complete binarizations} and will be shown to exist in abundance
as naturally arising from binary trees with $m$ leaves,
and Zhou-Bruck scheme is an instance of them.  The main tools in
our argument are the well-known leaf entropy theorem and a technical fact
which we call the structure lemma. 

Consider the following binary tree with 5 nodes and 6 leaves:
\begin{equation}\label{eq:example-tree}
\medskip
\raisebox{-.5\height}{\includegraphics[scale=0.8]{figs-3.mps}}
\end{equation}
The {\em leaf entropy theorem\/} states that, given a probability
distribution $\bp=\dist{p_0,\dots,p_5}$ on the leaves, the Shannon entropy
$H(\bp)$ is equal to the weighted sum $\sum_{i=1}^5 P_iH(\pi_i)$ of the
branching entropies $H(\pi_i)$ of the nodes, where the weight $P_i$ of node
$i$ is the sum of probabilities of the leaves under
it~\cite{massey83,DBLP:books/daglib/0016881,Knuth:art3}.
For example, $P_3=p_0+p_1+p_3+p_4$, and
$\pi_3=\dist{p_0+p_1+p_4,\,p_3}$.

As an interpretation of the theorem, consider a loaded dice $X$ with the
probability distribution $\bp$ of the 6 faces.  Each roll of $X$ generates,
according to the tree \eqref{eq:example-tree}, five possible coin tosses
$X_i$ with biases $\pi_i$, and $X_i$ has an output with probability $P_i$.
For example, if the dice roll $X$ is 1, then coins $X_1$, $X_3$, and $X_4$
give an output, as the tree is conveniently represented by squares (leaf,
dice roll) and circles (node, coin toss).  The leaf entropy theorem tells us
that the amount of information of the dice roll and the 5 coin tosses are the
same.  This suggests that $X_i$'s may be used as sources of randomness
to generate unbiased and independent random bits, possibly combined together,
at a rate as high as the entropy of $X$.

The mapping $X\mapsto (X_1,\dots,X_5)$ is a complete binarization: if $\Psi$
is 2-extracting, then $\Psi'(X)=\Psi(X_1)*\dots*\Psi(X_5)$ is 6-extracting.
Note that $X_i$'s are not independent.  However, $\Psi(X_i)$'s are
independent and therefore we can concatenate them.  Moreover, if $\Psi$ is
asymptotically optimal, then $\Psi'$ is also asymptotically optimal.  If one
or more of $X_i$'s are omitted, then the resulting $\Psi'$ is still
6-extracting, but not asymptotically optimal anymore.  And the same story
holds true of any binary tree.


\section{Extracting Procedures and Binarization}
\subsection{Extracting Procedures}
Our dice $X$ has $m$ faces with values $0, 1,\dots, m-1$ with probability
distribution $\dist{p_0,\dots,p_{m-1}}$.  A sequence
$x=x_1\dots x_n\in\{0,1,\dots,m-1\}^{n}$ is considered to be taken from $n$
repeated throws of the dice.  Summarized below are some necessary facts on
extracting procedures.  Refer to \cite{pae06-randomizing} and \cite{pae15}
for details.

\begin{definition}[\cite{peres92,pae06-randomizing}]\label{def:extracting}
A function $f\colon\{0,1,\dots,m-1\}^{n}\to \{0,1\}^\ast$ is {\em
  $m$-extracting} if for each pair $z_1,z_2$ in $\{0,1\}^\ast$ such that
$|z_1|=|z_2|$, we have $\Pr(f(x)=z_1)=\Pr(f(x)=z_2)$, regardless of
the distribution $\dist{p_0,\dots,p_{m-1}}$.
\end{definition}
\noindent
\begin{definition}
A function $\Psi\colon\{0,1,\dots,m-1\}^{*}\to \{0,1\}^\ast$ is called an
{\em $m$-extracting procedure} if its restriction on $\{0,1,\dots,m-1\}^n$ is
extracting, for every $n\ge0$.
\end{definition}
\noindent

Define $\Psi_1$ on $\{0,1\}^2$ by the rule \eqref{eq:vonNeumann} and call it
von Neumann function.  Extend it by, for an empty string,
\[
\Psi_1(\lambda)=\lambda,
\]
for a nonempty even-length input, 
\[
\Psi_1(x_1x_2\dots x_{2n})=\Psi_1(x_1x_2)*\cdots*\Psi_1(x_{2n-1}x_{2n}),
\]
where $*$ is concatenation, and for an odd-length input, drop the last bit
and take the remaining even-length bits.  Then the resulting function
$\Psi_1$ is a 2-extracting procedure.  Of course, there are more interesting
extracting procedures.  Asymptotically optimal 2-extracting procedures like
Elias's~\cite{elias72,pae-loui05,pae06-randomizing} and
Peres's~\cite{peres92,DBLP:journals/ipl/Pae13,pae15} also extend von Neumann
function but do not simply repeat it.

Denote by $S_{(n_0,n_1,\dots,n_{m-1})}$ the subset of
$\{0,1,\dots,m-1\}^{n}$ that consists of sequences with $n_i$ $i$'s.
Then
\[
\{0,1,\dots,m-1\}^{n}=\bigcup_{n_0+n_1+\dots+n_{m-1}=n} S_{(n_0,n_1,\dots,n_{m-1})},
\]
and each $S_{(n_0,n_1,\dots,n_{m-1})}$ is an {\em equiprobable\/} subset of
elements whose probability of occurrence is $p_0^{n_0}p_1^{n_1}\cdots
p_{m-1}^{n_{m-1}}$.  The size of an equiprobable set is given by a
multinomial coefficient like
\[
{n\choose{n_0,n_1,\dots,n_{m-1}}}=\frac{n!}{n_0! n_1!\cdots n_{m-1}!}.
\]
When $m=2$, an equiprobable set $S_{(l,k)}$ is also written as $S_{n,k}$,
where $n=l+k$, and its size can also be written as an equivalent binomial
coefficient as well as the multinomial one:
\[
{n\choose k}={n\choose{l,k}}.
\]

\newcommand\bN{\bf N} Extracting functions can be characterized using the
concept of {\em multiset}.  A multiset is a set with repeated elements;
formally, a multiset $M$ on a set $S$ is a pair $(S,\nu)$, where
$\nu\colon S\to\bN$ is a multiplicity function and $\nu(s)$ is called the
{\em multiplicity}, or the number of {\em occurrences} of $s\in S$.  The size
$|M|$ of $M=(S,\nu)$ is $\sum_{s\in S}\nu(s)$.  For multisets $A$ and $B$,
$A\uplus B$ is the multiset such that an element occurring $a$ times in $A$
and $b$ times in $B$ occurs $a+b$ times in $A\uplus B$.  So
$|A\uplus B|=|A|+|B|$, and the operation $\uplus$ is associative.

When we write $x\in M=(S,\nu)$, it simply means that $x\in S$.  However, when
we use the expression ``$x\in M$'' as an index, the multiplicity of the
elements is taken into account.  For example, for multisets $A$ and $B$, the
multiset $A\uplus B$ can be redefined as $\{x\mid x\in A \mbox{ or } x\in
B\}$.

By Definition~\ref{def:extracting}, the image of an extracting function
consists of multiple copies of $\{0,1\}^N$, the exact full set of binary
strings of various lengths $N$'s.  For example, von Neumann procedure defined
above sends $\{0,1\}^6$ to 12 copies of $\{0,1\}$, 6 copies $\{0,1\}^2$, and
one copy of $\{0,1\}^3$.
\begin{definition}[\cite{pae15}]
A multiset $A$ of bit strings is {\em extracting} if, for each $z$ that
occurs in $A$, all the bit strings of length $|z|$ occur in $A$ the same time
as $z$ occurs in $A$.
\end{definition} 

For multisets $A$ and $B$ of bit strings, define a new multiset
$A*B=\{s*t\mid s\in A, t\in B\}$, and this operation is associative, too. If
$A$ and $B$ are extracting, both $A*B$ and $A\uplus B$ are extracting.
Denote by $f\dpr{C}$ the multiset $\{f(x)\mid x\in C\}$, or equivalently,
$(f(C),\nu)$ with $\nu(z)=|f^{-1}(z)\cap C|$ for $z\in f(C)$.  Note that
$|f\dpr{C}|=|C|$.  For a disjoint union $C\cup D$, we have $f\dpr{C\cup
D}=f\dpr{C}\uplus f\dpr{D}$.
With this notation, $\Psi_1\dpr{\{0,1\}^6}=12\cdot\{0,1\}\uplus
6\cdot\{0,1\}^2\uplus1\cdot\{0,1\}^3$.

The following lemma reinterprets the definition of extracting function in
terms of equiprobable sets and their images.
\begin{lemma}[\cite{pae15}]\label{lemma:extracting-multiset}
A function $f\colon\{0,1,\dots,m-1\}^{n}\to \{0,1\}^\ast$ is extracting if
and only if $f\dpr{S_{(n_0,n_1,\dots,n_{m-1})}}$ is extracting for each tuple
$(n_0,n_1,\dots,n_{m-1})$ of nonnegative integers such that
$n_0+n_1+\cdots+n_{m-1}=n$.
\end{lemma}

\subsection{Binarization}
Given a function $\phi\colon\{0,1,\dots,m-1\}\to\{0,1,\lambda\}$, $\phi(X)$
is a Bernoulli random variable with distribution $\dist{p,q}$, where
\[
p=\sum_{\phi(i)=0}p_i/s,\; q=\sum_{\phi(i)=1}p_i/s,\,\text{and}\;
 s=\sum_{\phi(i)\not=\lambda}p_i.
\]
Extend $\phi$ to $\{0,1,\dots,m-1\}^n$, by letting, for $x=x_1\dots x_n$,
$\phi(x)=\phi(x_1)*\dots*\phi(x_n)$.  Then, for an equiprobable set
$S=S_{(n_0,\dots,n_{m-1})}$, its image under $\phi$ is also equiprobable,
that is,
\[
\phi(S)=S_{(l, k)},
\]
where
\[
l=\sum_{\phi(i)=0}n_i,\quad k=\sum_{\phi(i)=1}n_i.
\]

A {\em binarization\/} takes a sequence over $\{0,1,\dots,m-1\}$ and outputs
several binary sequences that are to be separately fed into a binary
extracting procedure and then concatenated together to obtain random bits.

\begin{definition} A collection of functions
$\Phi=\{\Phi_i:\{0,1,\dots,m-1\}\to\{0,1,\lambda\}\mid i=1,\dots,M\}$ is
called a {\em binarization} if, when extended to $\{0,1,\dots,m-1\}^n$, given
a 2-extracting procedure $\Psi$, the mapping
$x\mapsto\Psi'(x)=\Psi(\Phi_1(x))*\dots*\Psi(\Phi_M(x))$ is an $m$-extracting
function.  Here, each $\Phi_i$ is called a component of $\Phi$, and we often
regard $\Phi$ as a mapping on $\{0,1,\dots,m-1\}^*$ given by
$\Phi(x)=(\Phi_1(x),\dots,\Phi_M(x))$.  For an asymptotically optimal
2-extracting procedure $\Psi$, if the resulting $\Psi'$ is asymptotically
optimal, then $\Phi$ is called a {\em complete} binarization.
\end{definition}

\newcommand\supp{\mathrm{supp}}
Now, for a function $\phi\colon\{0,1,\dots,m-1\}\to\{0,1,\lambda\}$, let
\begin{align*}
\supp_0(\phi)&=\{x\mid \phi(x)=0\},\\
\supp_1(\phi)&=\{x\mid \phi(x)=1\},\\
\supp(\phi)&=\{x\mid \phi(x)\not=\lambda\}=\supp_0(\phi)\cup\supp_1(\phi),
\end{align*}
and call them 0-support, 1-support, and support of $\phi$, respectively. 
Call $\phi$ {\em degenerate\/} if its 0-support or 1-support is empty so
that $\phi(X)$ is a degenerate Bernoulli random variable.


\newcommand{\leaf}{\mathrm{leaf}}
Consider a binary tree with $m$ external nodes labeled uniquely with
$0,1,\dots,m-1$.  For an internal node $v$ define a function
$\phi_v\colon \{0,1,\dots,m-1\}\to\{0,1,\lambda\}$ as follows:
\[
\phi_v(x)=\begin{cases}
  0,&\text{if $x\in\leaf_0(v)$,}\\
  1,&\text{if $x\in\leaf_1(v)$,}\\
  \lambda,&\text{otherwise.}
 \end{cases}
\]
where $\leaf_0(v)$ ($\leaf_1(v)$, respectively) is the set of external nodes
on the left (right, respectively) subtree of $v$.  Since there are exactly
$m-1$ internal nodes, we uniquely name them with $1,\dots,{m-1}$, with 1 the
root node, and the corresponding functions $\Phi_1,\dots,\Phi_{m-1}$.  
Call such trees {\em $m$-binarization trees}.  

For example, the tree \eqref{eq:example-tree} that we considered in the
introduction is a 6-binarization tree and defines the following functions:
\begin{equation*}
\small
\medskip
\begin{tabular}{c|c|c|c|c|c}
\hline
$x$& $\Phi_1(x)$&$\Phi_2(x)$ & $\Phi_3(x)$ & $\Phi_4(x)$ & $\Phi_5(x)$\\
\hline
0&1&$\lambda$&0&1&1\\
1&1&$\lambda$&0&0&$\lambda$\\
2&0&0&$\lambda$&$\lambda$&$\lambda$\\
3&1&$\lambda$&1&$\lambda$&$\lambda$\\
4&1&$\lambda$&0&1&0\\
5&0&1&$\lambda$&$\lambda$&$\lambda$\\
\hline
\end{tabular}
\end{equation*}

\begin{theorem}\label{thm:construction}
For an $m$-binarization tree, the set of associated functions
$\Phi=\{\Phi_1,\dots,\Phi_{m-1}\}$ is a complete binarization.  Also, any
nonempty subset of $\Phi$ is a binarization.
\end{theorem}
For a proof, we use the leaf entropy theorem together with a technical lemma
that we call {\em Structure Lemma.}   The coin $X_i=\Phi_i(X)$
has an output with probability $P_i=\sum_{j\in\supp(\Phi_i)}p_j$, and its
distribution is $\pi_i=\dist{p,q}$, where
\[
p=\sum_{j\in\supp_0(\Phi_i)}p_j/P_i,\quad\;
q=\sum_{j\in\supp_1(\Phi_i)}p_j/P_i.
\]
Stated below is the leaf entropy theorem in our context of $m$-binarization
trees.
\begin{theorem}[Leaf Entropy Theorem]\label{thm:entropy}
The branching entropies of $\Phi_i(X)$ weighted by the probability $P_i$ sum
up to the entropy of $X$:
\[
H(X)=\sum_{i=1}^{m-1}P_i H(\pi_i).
\]
\end{theorem}
The following is the main technical tool of this work and we prove it in
Section~\ref{sec:proofs}.
\begin{lemma}[Structure Lemma]\label{lemma:structure}
Let $\Phi=\{\Phi_1,\dots,\Phi_{m-1}\}$ be the set of functions defined by an
$m$-binarization tree.  Then the mapping $\Phi\colon x\mapsto \Phi(x)=
(\Phi_1(x),\dots,\Phi_{m-1}(x))$ gives a one-to-one correspondence between
an equiprobable subset $S=S_{(n_0,n_1,\dots,n_{m-1})}$ and
$\Phi_1(S)\times\cdots\times\Phi_{m-1}(S)$.
\end{lemma}

\begin{proof}[Proof of Theorem~\ref{thm:construction}]
Let $\Psi$ be a 2-extracting procedure.  For an equiprobable set $S$, each
$S_i=\Phi_i(S)$ is equiprobable, and thus $\Psi\dpr{S_i}$ is extracting, by
Lemma~\ref{lemma:extracting-multiset}.  Now, by Lemma~\ref{lemma:structure},
$\Psi'\dpr{S}=\Psi\dpr{S_1}*\dots*\Psi\dpr{S_{m-1}}$.  Since each
$\Psi\dpr{S_i}$ is extracting, their concatenation $\Psi'\dpr{S}$ is extracting,
by the associativity of concatenation of multisets and the fact that
concatenation of extracting multisets is extracting.  The same holds true
even if we omit some components of $\Phi$.

Since the coin $X_i=\Phi_i(X)$ has the distribution $\pi_i$ and outputs with
the probability $P_i$, if $\Psi$ is asymptotically optimal, then the output
rate of $\Psi(X_i)$ converges to $P_iH(\pi_i)$ as the input size
$n\to\infty$.  Therefore, the output rate of $\Psi'$ approaches to
$\sum P_iH(\pi_i)$, which equals $H(X)$ by the leaf entropy theorem.
\end{proof}

\section{Examples}

\subsection{An Entropy-Preserving Binarization}\label{sec:new}
For a symbol $x\in\{0,1,\dots,m-1\}$ and $1\le i\le m-1$, consider
\[
x^{(i)}=\left\{
  \begin{split}
    0,\quad & x<i, \\
    1,\quad & x=i, \\
    \lambda,\quad & x>i.
  \end{split}
 \right.
\]
When $m=6$, we have their values as follow: 
\smallskip
\begin{equation*}
\small
\begin{tabular}{c|c|c|c|c|c|c}
\hline
$x$& $\Pr(x)$&$x^{(1)}$ & $x^{(2)}$ & $x^{(3)}$ & $x^{(4)}$& $x^{(5)}$\\
\hline
0&$p_0$&0&0&0&0&0\\
1&$p_1$&1&0&0&0&0\\
2&$p_2$&$\lambda$&1&0&0&0\\
3&$p_3$&$\lambda$&$\lambda$&1&0&0\\
4&$p_4$&$\lambda$&$\lambda$&$\lambda$&1&0\\
5&$p_5$&$\lambda$&$\lambda$&$\lambda$&$\lambda$&1\\
\hline
\end{tabular}\smallskip
\end{equation*}
These functions are associated with the following 6-binarization tree:
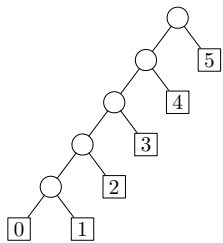
\begin{equation*}
\raisebox{-.5\height}{\includegraphics[scale=0.8]{figs-1.mps}}
\end{equation*}
For $x=x_1\dots x_n\in\{0,1,\dots,m-1\}^n$, define $x^{(i)}=x_1^{(i)}*\dots*
x_n^{(i)}$.  So for a sequence $x$ of length $n$, $x^{(i)}$ is a binary
sequence of length at most $n$.  For a binary extracting
procedure $\Psi$, the function $\Psi':\{0,1,\dots,m-1\}^n\to\{0,1\}^*$, defined
by
\[
\Psi'(x)=\Psi(x^{(1)})*\dots*\Psi(x^{(m-1)}),
\]
is $m$-extracting, and if $\Psi$ is asymptotically optimal, then so is $\Psi'$. 

To illustrate the structure lemma, for $m=4$, consider an equiprobable subset
$S=S_{(1,2,1)}\subset\{0,1,2\}^4$, and 
let $S^{(i)}=\{x^{(i)}\mid x\in S\}$.  Then, 
$S^{(i)}$ is another equiprobable set in $\{0,1\}^{n'}$.  For example, for
$S=S_{(1,2,1)}$, observe that
\[
\small
\medskip
\begin{split}
\begin{tabular}{c|c|c}
\hline
$x$&$x^{(2)}$&$x^{(1)}$\\
\hline
0112&0001&011\\
0121&0010&011\\
0211&0100&011\\
1012&0001&101\\
1021&0010&101\\
1102&0001&110\\
1120&0010&110\\
1201&0100&101\\
1210&0100&110\\
2011&1000&011\\
2101&1000&101\\
2110&1000&110\\
\hline
\end{tabular}
\end{split}
\] 
and we can see that, as multiset images of $x^{(1)}$ and $x^{(2)}$,
\[
\begin{split}
S^{\dpr{1}}&=4\cdot S_{(1,2)},\\
S^{\dpr{2}}&=3\cdot S_{(3,1)}.
\end{split}
\]
Note that 
\[
|S_{(1,2,1)}|=\frac{4!}{1!2!1!}=\frac{3!}{1!2!}\times\frac{4!}{3!1!}
=|S_{(1,2)}|\times|S_{(3,1)}|.
\]
Of course, by the structure lemma, $S$ is in one-to-one correspondence with
$S^{(1)}\times S^{(2)}$.

\subsection{Zhou-Bruck Binarization}\label{sec:zhou-bruck}
The following method was proposed by Zhou and
Bruck~\cite{DBLP:journals/corr/abs-1209-0726}.  For $x\in\{0,1,\dots,m-1\}$,
let $x'$ be the $\ceiling{\lg m}$-bit binary expansion of $x$, and also for
$\alpha\in\{0,1\}^*$, let
\[
x^\alpha=\begin{cases}
  a,&\text{if $\alpha a$ is a prefix of $x'$,}\\
  \lambda,&\text{otherwise.}
  \end{cases}
\]
That is, $x^\alpha$ is the bit that immediately follows $\alpha$ in the
standard binary expansion of $x$. 
For example, when $m=6$, we have the following functions:
\smallskip
\begin{equation*}
\small
\begin{tabular}{c|c|c|c|c|c|c|c}
\hline
$x$&$x'$&$x^\lambda$ & $x^0$ & $x^1$ & $x^{00}$& $x^{01}$ & $x^{10}$\\
\hline
0&000&0&0&$\lambda$&0&$\lambda$&$\lambda$\\
1&001&0&0&$\lambda$&1&$\lambda$&$\lambda$\\
2&010&0&1&$\lambda$&$\lambda$&0&$\lambda$\\
3&011&0&1&$\lambda$&$\lambda$&1&$\lambda$\\
4&100&1&$\lambda$&0&$\lambda$&$\lambda$&0\\
5&101&1&$\lambda$&0&$\lambda$&$\lambda$&1\\
\hline
\end{tabular}\smallskip
\end{equation*}
%
After the degenerate $x^1$ is removed, they are associated with the
following 6-binarization tree:
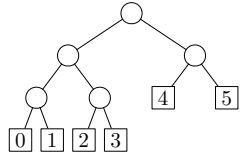
\begin{equation*}
\raisebox{-.5\height}{\includegraphics[scale=0.8]{figs-2.mps}}
\smallskip
\end{equation*}
The mapping $x\mapsto\Psi'(x)=\Psi(x^{\lambda})*\dots*\Psi(x^{1\dots1})$ is an
asymptotically optimal $m$-extracting procedure if $\Psi$ is asymptotically
optimal.



\section{The Structure Lemma}\label{sec:proofs}

Given a binarization tree and its subtree $T$, let $X_T$ be the restriction of
$X$ on the leaf set of $T$.  The leaf entropy theorem is proved by
induction using the following recursion,\footnote{Recall that a binary
  tree is recursively defined to be a set of nodes that is either an empty
  set (a terminal node), or consists of a root node, a left subtree and a
  right subtree, both of which are binary trees.}
\begin{equation}\label{eq:leaf-entropy-rec}
H(X_T)=\begin{cases}
0,& \text{if $T$ is a leaf,}\\
H(\pi)+pH(X_{T_1})+qH(X_{T_2}), &\text{otherwise,}
\end{cases}
\end{equation}
where, for nonempty $T$, $T_1$ and $T_2$ are the left and right
subtrees and $\pi=\dist{p,q}$ is the branching distribution of the root of
$T$.  The structure lemma holds for a similar reason.
\begin{proof}[Proof of Structure Lemma]
For an equiprobable subset $S=S_{(n_0,\dots,n_{m-1})}$ and a subtree $T$ of
the given binarization tree, let $S_T$ be the restriction of $S$ on the leaf
set of $T$.  Then we have a similar recursion
\begin{equation}\label{eq:structure-rec}
S_T\cong\begin{cases}
    \{0\},& \text{if $T$ is a leaf,}\\
    S_{(l,k)}\times S_{T_1}\times S_{T_2}, &\text{otherwise,}
  \end{cases}
\end{equation}
where, for nonempty $T$ and $\phi$ the branching function associated with the
root of $T$, $T_1$ and $T_2$ are the left and right subtrees and
\[
l=\sum_{\phi(i)=0}n_i,\quad k=\sum_{\phi(i)=1}n_i.
\]

First, if $T$ is a leaf with label $i$, then $S_T$ is a singleton set that
consists of a single string of $n_i$ $i$'s, hence the first part of
\eqref{eq:structure-rec}.  When $T$ is nonempty, the correspondence
$S_T\to S_{(l,k)}\times S_{T_1}\times S_{T_2}$ is given by
$x\mapsto (\phi(x),x_{T_1},x_{T_2})$, where $x_{T_1}$ and $x_{T_2}$ are
restrictions of $x$.  This correspondence is one-to-one because $\phi(x)$
encodes the branching with which $x$ is recovered from $x_{T_1}$ and
$x_{T_2}$, giving an inverse mapping
$S_{(l,k)}\times S_{T_1}\times S_{T_2}\to S_T$.  For example, consider tree
\eqref{eq:example-tree} and suppose that $T$ is the subtree rooted at the
node {\it 3}.  For $x=102235315401$, the following shows the restrictions of
$x$ and $\Phi_i(x)$'s.
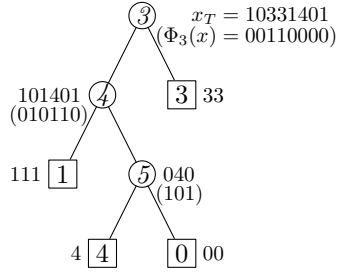
\[
\medskip
\raisebox{-.5\height}{\includegraphics[scale=1]{figs-300.mps}}
\]
By taking symbols one by one from $x_{T_1}=101401$ and $x_{T_2}=33$, according
to $\Phi_3(x)=00110000=(b_i)_{i=1}^8$, if $b_i$ is 0, from $x_{T_1}$,
otherwise, from $x_{T_2}$, we recover $x_T=10331401$.

Induction on subtrees proves the lemma.
\end{proof}
See \cite{DBLP:conf/isit/Pae16} for an alternative proof.



\section{Remarks}
\subsection{Leaf Entropy Theorem and Structure Lemma}
The leaf entropy theorem is well known in the information theory, and it
follows from the grouping rule of entropy (see, e.g., the defining property 3
of entropy in Shannon's original work \cite[p. 49]{shannon64mathematical}, or
Problem 2.27 of \cite{DBLP:books/daglib/0016881}), which is essentially the
recursion \eqref{eq:leaf-entropy-rec} in Section \ref{sec:proofs}.  As we
saw, the structure lemma is proved similarly, hinting that they are closely
related.  In fact, using the asymptotic equipartition property (AEP)
\cite{DBLP:books/daglib/0016881}, the structure lemma implies the leaf
entropy theorem.

For a large $n$, the typical set $A^{(n)}$ consists of
$x=(x_1,x_2,\dots,x_n)$ that contains about $n_0=p_0n$ 0's, $n_1=p_1n$ 1's,
$\dots$, $n_{m-1}=p_{m-1}n$ $(m-1)$'s.  Let $S=S_{(n_0,\dots,n_{m-1})}$. The
asymptotic equipartition property implies that
$\lim_{n\to\infty}\frac1n\log|S|=H(X)$.  On the other hand, by Structure
Lemma, $S=S_1\times\dots\times S_{m-1}$, where $S_i=\Phi_i(S).$
Note that $S_i=S_{(l_i,k_i)}$, where
\[
l_i=\sum_{j\in\supp_0(\Phi_i)}n_j,\quad
k_i=\sum_{j\in\supp_1(\Phi_i)}n_j,
\]
and $(l_i+k_i)/n\to P_i$ and $\dist{l_i/n, k_i/n}\to\pi_i$ as $n\to\infty$.
Since $\frac1{(l_i+k_i)}\log|S_i|\to H(\pi_i)$, we have
\[
\frac1n\log|S_i|\to P_iH(\pi_i),
\]
and
\[
  \frac1n\log|S|
      =\frac1n\sum_{i=1}^{m-1}\log|S_i|
      \to\sum_{i=1}^{m-1}P_iH(\pi_i),
\]
as $n\to\infty$.

\subsection{Generalization of Structure Lemma to Non-Binary Trees}
The leaf entropy theorem holds for general trees.  The structure lemma also
can be generalized to trees whose nodes are not necessarily of degree 2 and
whose leaves have unique labels, although in that case, the naming
``binarization tree'' might not be appropriate.  

\subsection{$m$-ary Asymptotically Optimal Extracting Algorithm}
As an immediate application, take the original binary Peres procedure $\Psi$
and apply Theorem~\ref{thm:construction}.  The resulting $\Psi'$ is an
$m$-ary asymptotically optimal extracting procedure.  As with the original
Peres algorithm and its generalization, $\Psi'$ runs in $O(n\log n)$ time,
for a fixed $m$, because $\Phi_i(x)$ is computed in linear time and
$|\Phi_i(x)|\le n$ for each $i$.

\subsection{Other Applications of Binarization Trees}
Peres algorithm is a simple extracting algorithm defined recursively using
the famous von Neumann trick as a base, whose output rate approaches the
information-theoretic upper bound~\cite{peres92}.  However, it is relatively
hard to explain why it works, and it appears partly due to this difficulty
that its generalization to many-valued source was discovered only
recently~\cite{pae15}.  Binarization tree provides a new unified way to
understand the original Peres algorithm and its generalizations and
facilitates finding many new Peres-style recursive algorithms~\cite{pae18p}.
By coming up with an appropriate binarization tree (not necessarily based on
binary tree but possibly a general tree), a Peres-style recursion follows.
As with our main result, Theorem~\ref{thm:construction}, the Peres-style
recursive algorithms are extracting by the corresponding structure lemma, and
asymptotically optimal by the leaf entropy theorem.

The structure lemma gives many different ways to factorize a set of
$m$-combinations into sets of binary combinations.  We can use this idea to
give a ranking on $m$-combinations, which can be seen as a mixed-radix number
system whose radices are binomial numbers~\cite{pae18r}.

\subsection{Binarization Trees and DDG-trees}
DDG-trees (discrete distribution generation trees) work in the opposite way
of binarization
trees~\cite{knuth-yao76,han-hoshi97,pae-loui05,pae06-randomizing}.  With a
binarization tree, the leaves correspond to the source and various coins are
produced.  With DDG trees, the nodes correspond to
the source and target symbols of the leaves are produced.  However, the
essential difference is that DDG has the same branching distribution for
every node and that the leaves don't have to have unique labels.  If the
various {\em source coins\/} with distributions $\pi_i$'s are provided, and
the coins are tossed starting from the root in the fashion of DDG-trees, then
we arrive at leaves with the target probability distribution
$\dist{p_0,\dots,p_{m-1}}$.  Therefore, binarization tree can be regarded as
a generalization of DDG-tree with more than one source and unique labels on
leaves.



\end{document}